\documentclass[a4paper,10pt,aps,pra,notitlepage,superscriptaddress,nofootinbib,tightenlines]{revtex4-1}

\usepackage{graphicx}
\usepackage{amsmath}
\usepackage{amsfonts} 
\usepackage{amssymb}
\usepackage{color}
\usepackage{amsthm}
\usepackage{dsfont} 
\usepackage{fullpage}

\newcommand{\D}{\widetilde{D}}

\DeclareMathOperator{\supp}{supp}
\DeclareMathOperator{\Tr}{Tr}

\newtheorem{lemma}{Lemma}
\newtheorem{proposition}[lemma]{Proposition}
\newtheorem{theorem}[lemma]{Theorem}

\newtheorem{definition}[lemma]{Definition}
\newtheorem{example}{Example}

\begin{document}

\title{\Large Investigating Properties of a Family of Quantum R\'enyi Divergences}

\author{Simon M.~Lin}
\affiliation{Centre for Quantum Technologies, National University of Singapore, Singapore 117543, Singapore}
\email{email: a0067236@nus.edu.sg}
\author{Marco Tomamichel}
\affiliation{School of Physics, The University of Sydney, Sydney 2006, Australia}
\affiliation{Centre for Quantum Technologies, National University of Singapore, Singapore 117543, Singapore}

\begin{abstract}
Audenaert and Datta recently introduced a two-parameter family of relative R\'{e}nyi entropies, known as the $\alpha$-$z$-relative R\'{e}nyi entropies. The definition of the $\alpha$-$z$-relative R\'{e}nyi entropy unifies all previously proposed definitions of the quantum R\'{e}nyi divergence of order $\alpha$ under a common framework. Here we will prove that the $\alpha$-$z$-relative R\'{e}nyi entropies are a proper generalization of the quantum relative entropy by computing the limit of the $\alpha$-$z$ divergence as $\alpha$ approaches one and $z$ is an arbitrary function of $\alpha$. We also show that certain operationally relevant families of R\'enyi divergences are differentiable at $\alpha = 1$. Finally, our analysis reveals that the derivative at $\alpha = 1$ evaluates to half the relative entropy variance, a quantity that has attained operational significance in second-order quantum hypothesis testing and channel coding for finite block lengths.
\end{abstract}

\maketitle

\section{Introduction}

In classical information theory, the Shannon entropy \cite{Shannon} measures the average information content of a random variable, while the Kullback-Leibler (KL) divergence \cite{KL} (also known as relative entropy) is a non-symmetric measure of the difference between probability distributions. 

From a mathematical perspective, the Shannon entropy \cite{Shannon} is uniquely characterized by a set of axioms, that is, the Shannon entropy satisfies the axioms, and the only functional on the set of probability distributions that satisfies the axioms is the Shannon entropy. By relaxing one of the axioms, we obtain a family of functionals on the set of probability distributions indexed by a parameter $\alpha$, called the R\'{e}nyi entropies of order $\alpha$ \cite{Renyi}. In particular, when $\alpha$ tends to $1$, we recover our usual notion of the Shannon entropy \cite{Renyi}. In a similar spirit, we also have that the KL divergence is uniquely characterized by a set of axioms, and that we may also obtain a family of functions on the set of pairs of probability distributions indexed by a parameter $\alpha$, called the R\'{e}nyi divergences of order $\alpha$ \cite{Renyi}. In particular, when $\alpha$ tends to $1$, we recover our usual notion of the KL divergence \cite{Renyi}. 

The families of R\'{e}nyi entropies of order $\alpha$ and R\'{e}nyi divergences of order $\alpha$ contain several other entropies and divergences of operational significance, which are used e.g.\ in cryptography and information theory. For instance, when $\alpha$ tends to infinity, we obtain the minimum entropy and maximum relative entropy, where the minimum entropy quantifies how hard it is to guess a random variable. More generally, the R\'{e}nyi entropies and divergences are a prominent tool in information theory (see, e.g.~\cite{Csiszar}), and have an operational interpretation, for example in hypothesis testing or channel coding.

A natural question to ask is if the notion of R\'enyi divergence has a natural extension in the quantum setting. Indeed, by replacing the notion of probability distributions with the notion of positive semi-definite operators on a finite-dimensional Hilbert space, we could obtain the quantum counterparts of Shannon entropy and KL divergence, namely the von Neumann entropy~\cite{vN}, and the quantum relative entropy~\cite{Umegaki}. 

So far, several definitions for the quantum R\'{e}nyi divergences of order $\alpha$ have been proposed. The most widely used definition is based on Petz' quasi-entropies (see, e.g.~\cite{petz}). More recently a different definition has been proposed independently by M\"{u}ller-Lennert \emph{et al.}~\cite{ML} and Wilde \emph{et al.}~\cite{WWY}. Both the initial and the new definitions have found operational significance in quantum hypothesis testing (see~\cite{MO} and references therein). They are related by duality relations for conditional R\'enyi entropies~\cite{tomamichel13}.
 The new quantum generalization has also lead to significant progress in determining the strong converse property of certain quantum channels for classical~\cite{WWY}, quantum~\cite{tomamichelww14}, and entanglement-assisted~\cite{gupta13} information transmission.

The non-uniqueness of the definition is a direct result of the non-commutati\-vity of general quantum states and it is possible to devise various different generalizations of the classical R\'enyi divergence beyond the ones mentioned above.
Most recently, a two-parameter family for the quantum R\'{e}nyi divergence, namely the $\alpha$-$z$ relative R\'{e}nyi entropies, has been proposed by Audenaert and Datta \cite{AD}, in the hope that it could shed some light on the original problem of the definition of the quantum R\'{e}nyi divergence of order $\alpha$. 
In this work we investigate some properties of the general $\alpha$-$z$ entropies, in particular we establish that they are continuous and differentiable at $\alpha = 1$ under some weak conditions. Furthermore, we compute the derivative at $\alpha = 1$ and find that it corresponds to the quantum relative information variance that has recently found operational significance in second-order quantum hypothesis testing~\cite{TH13,LI13} and quantum channel coding~\cite{tomamicheltan14}.

{The remainder of this paper is organized as follows.
Section~\ref{sec:pre} introduces the relevant notation and definitions of the R\'enyi divergences used in this work. Section~\ref{sec:results} then provides an overview and discussion of the main technical results. Their proofs are provided in Section~\ref{sec:proofs}.}

\section{Notations and Preliminaries}
\label{sec:pre}

\subsection{Notation}

Throughout this paper, let $\mathcal{H}$ denote a finite-dimensional Hilbert space, which is a finite-dimensional vector space over the field of complex numbers, endowed with a complex inner product. Let $\mathcal{L}(\mathcal{H})$, $\mathcal{T}(\mathcal{H})$, $\mathcal{P}(\mathcal{H})$, and $\mathcal{P}^+(\mathcal{H})$ denote the set of linear, Hermitian, positive semi-definite, and positive definite operators on $\mathcal{H}$, respectively. Finally, we use $\mathcal{D}(\mathcal{H}) := \{ \rho \in \mathcal{P}(\mathcal{H}) : \Tr(\rho) = 1 \}$  to denote the set of all density operators on $\mathcal{H}$.

We note that a linear operator $\rho\in\mathcal{L}(\mathcal{H})$ is Hermitian if and only if it is unitarily diagonalizable, and its eigenvalues are real. A Hermitian operator $\rho\in\mathcal{T}(\mathcal{H})$ is positive semi-definite, denoted $\rho \geq 0$, (respectively positive definite, denoted $\rho > 0$) if and only if its eigenvalues are non-negative (respectively positive).  
For all $\rho, \sigma\in\mathcal{P}(\mathcal{H})$, we write $\rho\geq\sigma$ or $\sigma\leq\rho$ if and only if $\rho-\sigma\geq0$. 

For a given $\rho\in\mathcal{L}(\mathcal{H})$, we denote the kernel and support of $\rho$ by $\ker\rho$ and $\supp \rho$, respectively. For a given $\rho\in\mathcal{T}(\mathcal{H})$ and $v\in\mathcal{H}$, we have $v\in\ker\rho$ if and only if $v$ is a eigenvector of $\rho$ associated to the zero eigenvalue. Moreover, $\mathcal{H}=\ker\rho\oplus\supp\rho$.
For any $\sigma,\tau\in\mathcal{T}(\mathcal{H})$, we use the notation $\sigma\gg\tau$ to indicate the fact that $\ker\sigma\subseteq\ker\tau$ (or equivalently, $\supp\tau\subseteq\supp\sigma$), and we use the notation $\sigma\perp\tau$ to indicate that $\sigma$ and $\tau$ are orthogonal to each other. Also, for any $\rho\in\mathcal{T}(\mathcal{H})$, we denote the projection operator of $\rho$ onto $\supp\rho$ by $\Pi_{\rho}$. Finally, we  denote the Schatten $p$-norm on $\mathcal{L}(\mathcal{H})$, where $p\in[1,\infty]$, by $\|\cdot\|_{p}$. 

For positive semi-definite operators $\rho \in \mathcal{P}(\mathcal{H})$, we write $\rho^{-1}$ for the generalized inverse and $\log \rho$ for the logarithm evaluated on the support of $\rho$.

\subsection{Classical Divergence}

First, let us recall the definition of the KL divergence, which is a non-symmetric measure of the difference between probability distributions~\cite{KL}. Let $X$ and $Y$ be random variables with alphabet $\{a_1, a_2, \cdots, a_n\}$, and let $p(a_i)$ and $q(a_i)$ denote the probabilities that the outcome $a_i$ occurs in $X$ and $Y$ respectively.
The {KL divergence} of $Y$ from $X$ is defined as
$$D(X\|Y) :=\sum_{i=1}^n p(a_i)\log \frac{p(a_i)}{q(a_i)}.$$

In 1961, R\'{e}nyi produced a set of axioms that uniquely characterize the KL divergence. By relaxing one of the axioms, R\'{e}nyi further characterized a one-parameter family of divergences, generalizing the KL divergence. 
The {R\'{e}nyi divergence of order $\alpha$ of $Y$ from $X$}, where $\alpha\in(0,1)\cup(1,\infty)$, is defined as~\cite{Renyi}
$$D_{\alpha}(X\| Y) :=\frac{1}{\alpha-1}\log\left(\sum_{i=1}^n p(x_i)^{\alpha}q(x_i)^{1-\alpha}\right).$$
It is well known that $\lim_{\alpha\to1} D_{\alpha}(X\| Y)=D(X\| Y)$.

\subsection{Quantum Divergences}

The R\'{e}nyi divergence has a natural extension to the quantum setting.

\begin{definition}\label{4} Let $\rho \in \mathcal{D}(\mathcal{H})$, $\sigma\in\mathcal{P}(\mathcal{H})$. The \textbf{quantum R\'{e}nyi divergence of order $\alpha$}, where $\alpha\in(0,1)\cup(1,\infty)$, of $\sigma$ from $\rho$ is defined as
\begin{eqnarray*}
D_{\alpha}(\rho\|\sigma) :=  \left\lbrace
\begin{array}{l l}
\frac{1}{\alpha-1}\log {\Tr(\rho^{\alpha}\sigma^{1-\alpha})} & \quad \text{if } \sigma\gg\rho, \text{ or } \alpha <1 \text{ and } \rho\not\perp\sigma,\\
\infty & \quad \text{otherwise}
\end{array}
\right.
\end{eqnarray*}
\end{definition}
As we will see, in the limit $\alpha \to 1$ this definition yields the quantum relative entropy,
$D(\rho\|\sigma) := \Tr(\rho(\log\rho-\log\sigma))$ if $\sigma \gg \rho$.

M\"{u}ller-Lennert \emph{et al.} \cite{ML} and independently, Wilde \emph{et al.} \cite{WWY}, proposed a new definition for the quantum R\'{e}nyi divergence of order $\alpha$:\footnote{In the latter paper, the divergence is called ``sandwiched R\'enyi relative entropy''.}

\begin{definition}\label{6} Let $\rho \in \mathcal{D}(\mathcal{H})$, $\sigma\in\mathcal{P}(\mathcal{H})$. The \textbf{sandwiched quantum R\'{e}nyi divergence of order $\alpha$}, where $\alpha\in(0,1)\cup(1,\infty)$, of $\sigma$ from $\rho$ is defined as
\begin{eqnarray*}
\D_{\alpha}(\rho\|\sigma) := \left\lbrace
\begin{array}{l l}
\frac{1}{\alpha-1}\log\Tr\left(\left(\sigma^{\frac{1-\alpha}{2\alpha}}\rho\sigma^{\frac{1-\alpha}{2\alpha}}\right)^{\alpha}\right) & \quad \text{if } \sigma\gg\rho, \text{ or } \alpha <1 \text{ and } \rho\not\perp\sigma\\
\infty & \quad \text{otherwise}
\end{array}
\right.
\end{eqnarray*}
\end{definition}

Here, we note that when $\rho$ and $\sigma$ commute, the two definitions for the quantum R\'{e}nyi divergence of order $\alpha$, namely $D_{\alpha}(\rho\|\sigma)$ and $\D_{\alpha}(\rho\|\sigma)$, coincide. 
 
Note that both proposed definitions for the quantum R\'{e}nyi divergence of order $\alpha$ have found operational significances, and thus it is not sufficient to consider just a single quantum generalization of the R\'enyi divergence. Moreover, none of the two generalizations possesses all the desired limits. For instance, the relative max-, and collision entropies are shown to be not a specialization of $D_{\alpha}(\rho\|\sigma)$ for any value of $\alpha$ \cite{CRT,Tomamichel}, while the relative min entropy is shown not to be a specialization of $\D_{\alpha}(\rho\|\sigma)$ for any value of $\alpha$ \cite{DL}.  

Moreover, Mosonyi and Ogawa proposed a piecewise definition for the quantum R\'{e}nyi divergence of order $\alpha$, based on its operational interpretation in quantum hypothesis testing~\cite{MO}:
\begin{eqnarray*}
\bar{D}_{\alpha}(\rho\|\sigma) := \left\lbrace
\begin{array}{l l}
D_{\alpha}(\rho\|\sigma) & \quad \text{if }\alpha<1,\\
\D_{\alpha}(\rho\|\sigma) & \quad \text{if }\alpha>1
\end{array}
\right.
\end{eqnarray*} 

The definition $\bar{D}_{\alpha}(\rho\|\sigma)$ for the quantum R\'{e}nyi divergence of order $\alpha$ satisfies the data processing inequality for all admissible values of $\alpha$ \cite{ML,Beigi,FL}. Furthermore, this definition for the quantum R\'{e}nyi divergence of order $\alpha$ satisfies the axioms proposed by M\"uller-Lennert \emph{et al.} for the quantum R\'{e}nyi divergence of order $\alpha$ for all $\alpha\in\left(0,1\right)\cup(1,\infty)$.

\subsection{$\alpha$-$z$ relative R\'{e}nyi entropies}

More recently, Audenaert and Datta proposed a two-parameter family of $\alpha$-$z$ relative R\'{e}nyi entropies \cite{AD}, and along with the definition, proved some limiting properties of the $\alpha$-$z$ relative R\'{e}nyi entropies. The aim of defining a two-parameter family of $\alpha$-$z$ relative R\'{e}nyi entropies is to unite the different proposed definitions for the R\'{e}nyi divergences.

\begin{definition}\label{10}
Let $\rho\in \mathcal{D}(\mathcal{H})$, $\sigma\in \mathcal{P}(\mathcal{H})$ with $\sigma\gg\rho$. Then for all $\alpha\in\mathbb{R}\setminus\{1\}$ and $z\in\mathbb{R}\setminus\{0\}$, the \textbf{$\alpha$-$z$ relative R\'{e}nyi entropy} of $\sigma$ from $\rho$ is defined as
$$D_{\alpha,z}(\rho\|\sigma):=\frac{1}{\alpha-1}\log\Tr\left(\left(\sigma^{\frac{1-\alpha}{2z}}\rho^{\frac{\alpha}{z}}\sigma^{\frac{1-\alpha}{2z}}\right)^z\right).$$ 
\end{definition}

Hereafter, we will refer the $\alpha$-$z$ relative R\'{e}nyi entropies as $\alpha$-$z$ divergences for convenience.
For $z=1$ and $z=\alpha$, we see that the quantities $\D_{\alpha}(\rho\|\sigma)$ and $D_{\alpha}(\rho\|\sigma)$ are specializations of the $\alpha$-$z$ divergences.
\begin{align*}
  D_{\alpha,1}(\rho\|\sigma) &= D_{\alpha}(\rho\|\sigma) = \frac{1}{\alpha-1}\log\Tr\left(\sigma^{1-\alpha}\rho^{\alpha}\right), \quad\text{ and} \\
  D_{\alpha,\alpha}(\rho\|\sigma) &= \D_{\alpha}(\rho\|\sigma)=\frac{1}{\alpha-1}\log\Tr\left(\left(\sigma^{\frac{1-\alpha}{2\alpha}}\rho\sigma^{\frac{1-\alpha}{2\alpha}}\right)^{\alpha}\right).
\end{align*}

\section{Main Results and Discussion}
\label{sec:results}

{
Here we present our main results. Their proofs, together with various auxiliary results, are presented in the following section.
Our first result establishes that the $\alpha$-$z$ divergence is a generalization of the relative entropy in the following strong sense:
}

{
\begin{theorem}\label{31}
Let $\rho\in \mathcal{D}(\mathcal{H})$ and $\sigma\in \mathcal{P}(\mathcal{H})$ with $\sigma\gg\rho$. Suppose $J$ is an open interval containing $1$, and $g:J\to\mathbb{R}$ is a continuously differentiable function with $g(1)\neq0$. Then, $$\lim_{\alpha\to1}D_{\alpha,g(\alpha)}(\rho\|\sigma)=D(\rho\|\sigma)=\Tr(\rho(\log\rho-\log\sigma)).$$
\end{theorem}
}
Henceforth, we define $D_{1,z}(\rho\|\sigma):=\lim_{\alpha\to1} D_{\alpha,z}(\rho\|\sigma)=D(\rho\|\sigma)$ for all $z\neq0$.
In particular, our proof generalizes the arguments in~\cite{ML,WWY} which established that $\lim_{\alpha\to1} \D_{\alpha}(\rho\|\sigma) = D(\rho\|\sigma)$. 
We also remark that we have $$\lim_{\alpha\to1}\lim_{z\to\infty} D_{\alpha,z}(\rho\|\sigma)=D(\rho\|\sigma)=\Tr(\rho(\log\rho-\log\sigma))$$ in general, where this limit has been proved by Audenaert and Datta \cite{AD}. 
Finally, we note that the condition $g(1)\neq0$ is crucial in proving that $$\lim_{\alpha\to1}D_{\alpha,g(\alpha)}(\rho\|\sigma)=D(\rho\|\sigma)=\Tr(\rho(\log\rho-\log\sigma)).$$ Indeed, the limit is not reproduced in general when $g(1)=0$ (see \cite[Theorem 2]{AD} for further details).

Our second main result establishes both $\alpha \mapsto D_{\alpha}(\rho\|\sigma)$ and $\alpha \mapsto \D_{\alpha}(\rho\|\sigma)$ are continuously differentiable at $\alpha = 1$. Moreover, their derivatives agree and are proportional to the relative entropy variance.

\begin{theorem}\label{newthm}
Let $\rho\in \mathcal{D}(\mathcal{H})$ and $\sigma\in \mathcal{P}(\mathcal{H})$ with $\sigma\gg\rho$. Then $D_{\alpha,1}(\rho\|\sigma)$ and $D_{\alpha,\alpha}(\rho\|\sigma)$ are differentiable at $\alpha=1$, and we have 
\begin{align*}
  &\frac{\rm d}{{\rm d}\alpha}D_{\alpha,1}(\rho\|\sigma)\, \bigg|_{\alpha=1} = \frac{\rm d}{{\rm d}\alpha}D_{\alpha,\alpha}(\rho\|\sigma)\, \bigg|_{\alpha=1} = \frac12 V(\rho\|\sigma), \quad \textrm{where}\\
  &\qquad V(\rho\|\sigma) := \Tr\left(\rho(\log\rho-\log\sigma)^2\right)-(\Tr\left(\rho(\log\rho-\log\sigma)\right))^2
\end{align*}
is the relative entropy variance~\cite{TH13,LI13}.
\end{theorem}

The second result thus in particular establishes that the piecewise definition of Mosonyi and Ogawa, $\bar{D}_\alpha(\rho\|\sigma)$, is continuously differentiable in $\alpha$. 
An interesting question to ask is if the function $\alpha\mapsto\bar{D}_{\alpha}(\rho\|\sigma)$ is in fact smooth. However, the answer to this question is negative in general. 
In particular, the second derivative of $\bar{D}_{\alpha}(\rho\|\sigma)$ need not exist at $\alpha=1$ in general. 

Let us consider the following example.
\begin{example}
Let $\rho=\frac12\tiny\begin{pmatrix}1 & 1\\ 1 & 1\end{pmatrix}$ and $\sigma=\tiny\begin{pmatrix}p & 0\\ 0 & 1-p\end{pmatrix}$, where $p\in(0,1)\setminus\left\{\frac{1}{2}\right\}$. Then,
$$D_{\alpha,z}(\rho\|\sigma)=\frac{z}{\alpha-1}\log\frac{p^{\frac{1-\alpha}{z}}+(1-p)^{\frac{1-\alpha}{z}}}{2}.$$
From here, we verify via l'H\^{o}pital's rule that 
\begin{eqnarray*}
\frac{\rm{d}^2}{\rm{d}\alpha^2}D_{\alpha,1}(\rho\|\sigma)\bigg|_{\alpha=1} &=& 0, \qquad \text{ and}\\
\frac{\rm{d}^2}{\rm{d}\alpha^2}D_{\alpha,\alpha}(\rho\|\sigma)\bigg|_{\alpha=1} &=& -\frac{1}{4}(\log p-\log(1-p))^2.
\end{eqnarray*}
As the two quantities are different from each other, this shows that the second derivative of $\bar{D}_{\alpha}(\rho\|\sigma)$ does not exist at $\alpha=1$.
\end{example}

\section{Proofs}
\label{sec:proofs}

\subsection{Partial Derivative of the $\alpha$-$z$ Divergence with Respect to $z$}

It is generally well known that the functions $\alpha\mapsto D_{\alpha}(\rho\|\sigma)$ and $\alpha\mapsto \D_{\alpha}(\rho\|\sigma)$ are monotonically increasing \cite{Tomamichel, ML}. A natural question to ask then is whether if the $\alpha$-$z$ divergences are monotone with respect to $\alpha$ or $z$. The following follows by a simple application of the Araki-Lieb-Thirring trace inequality~\cite{LT,Araki}.

\begin{proposition}\label{13}
Let $\rho\in \mathcal{D}(\mathcal{H})$, $\sigma\in \mathcal{P}(\mathcal{H})$ with $\sigma\gg\rho$, and let $\alpha\in\mathbb{R}\setminus\{1\}$ and $z > 0$. Then, the function $z\mapsto D_{\alpha,z}(\rho\|\sigma)$ is monotonically decreasing if $\alpha>1$, and monotonically increasing if $\alpha<1$.
\end{proposition}

In general, we remark that proving monotonicity properties for the $\alpha$-$z$ divergences with respect to $\alpha$ is much harder than proving monotonicity properties for the $\alpha$-$z$ divergences with respect to $z$. However, numerical results lead us to conjecture that the $\alpha$-$z$ divergence is monotone with respect to $\alpha$ for a fixed $z$. Nevertheless, we shall prove some local monotonicity properties for the $\alpha$-$z$ divergences with respect to $\alpha$ at $\alpha = 1$ in the following subsections.

\subsection{Continuity of $\alpha$-$z$ Divergence as $\alpha\to 1$}

The main purpose of this section is to prove Theorem~\ref{31}. 
We will need the following technical ingredient. 
The proof follows from standard arguments (see, e.g.,~\cite[Lemma III.1]{mosonyi14}). It is omitted here to streamline the presentation.

\begin{lemma}\label{19}\label{29}
Let $J$ be an open interval in $\mathbb{R}$, and let $F:J\to \mathcal{P}^+(\mathcal{H})$ be a continuously differentiable function. Then, for all $z \in \mathbb{R}$, the function $x\mapsto\Tr(F(x)^z)$ is differentiable, and 
$$\frac{\rm d}{{\rm d}x} \Tr(F(x)^z) = \Tr(z(F(x))^{z-1}F'(x)) .$$
Furthermore, the function $x\mapsto \Tr(F(x)^x)$ is continuously differentiable, and 
$$\frac{\rm d}{{\rm d}x} \Tr(F(x)^x)=\Tr(F(x)^x\log F(x)+xF'(x)F(x)^{x-1}).$$
\end{lemma}

As an immediate consequence of the latter statement, we find the following:

\begin{proposition}\label{30}
For all $\rho\in \mathcal{D}(\mathcal{H})$ and $\sigma\in \mathcal{P}(\mathcal{H})$ with $\sigma\gg\rho$, and $z\in\mathbb{R}\setminus\{0\}$, we have $$\lim_{\alpha\to1}\frac{\partial}{\partial z}\Tr\left(\sigma^{\frac{1-\alpha}{2z}}\rho^{\frac{\alpha}{z}}\sigma^{\frac{1-\alpha}{2z}}\right)^z=0.$$
\end{proposition}

\begin{proof}
Note that the function $F: z\mapsto \sigma^{\frac{1-\alpha}{2z}}\rho^{\frac{\alpha}{z}}\sigma^{\frac{1-\alpha}{2z}}$ is continuously differentiable. Indeed, we have
\begin{eqnarray*}
&&F'(z) = -\frac{1-\alpha}{2z^2}\sigma^{\frac{1-\alpha}{2z}}(\log\sigma)\rho^{\frac{\alpha}{z}}\sigma^{\frac{1-\alpha}{2z}} \\
&& \qquad -\frac{\alpha}{z^2}\sigma^{\frac{1-\alpha}{2z}}(\log\rho)\rho^{\frac{\alpha}{z}}\sigma^{\frac{1-\alpha}{2z}}
-\frac{1-\alpha}{2z^2}\sigma^{\frac{1-\alpha}{2z}}\rho^{\frac{\alpha}{z}}\sigma^{\frac{1-\alpha}{2z}}(\log\sigma) ,
\end{eqnarray*}
which is continuously differentiable in $\alpha$.
Hence, at $\alpha=1$, we have 
$F(z) = \rho^{\frac{1}{z}}$ and $F'(z) = -\frac{1}{z^2}\rho^{\frac{1}{z}}\log\rho$.
Thus, Lemma \ref{29} yields  
\begin{eqnarray*}
\frac{\partial}{\partial z}\Tr \left( F(z)^z \right)
&=&\Tr\left(\rho\log\rho^{\frac{1}{z}}-z\left(\frac{1}{z^2}\rho^{\frac{1}{z}}\log\rho\right)\rho^{\frac{z-1}{z}}\right)\\
&=&\Tr\left(\rho\log\rho^{\frac{1}{z}}-\frac{1}{z}\rho\log\rho\right)
=0.
\end{eqnarray*}
\qed
\end{proof}

Now, we are ready to prove the following special case of Theorem~\ref{31}:

\begin{proposition}\label{20}
Let us fix a $\rho\in \mathcal{D}(\mathcal{H})$ and $\sigma\in \mathcal{P}(\mathcal{H})$ with $\sigma\gg\rho$. Then for all $z\neq0$, we have $\lim_{\alpha\to1}D_{\alpha,z}(\rho\|\sigma)=D(\rho\|\sigma)$.
\end{proposition} 

\begin{proof}
Let us define $F(\alpha)=\sigma^{\frac{1-\alpha}{2z}}\rho^{\frac{\alpha}{z}}\sigma^{\frac{1-\alpha}{2z}}$. By Lemma \ref{19}, the function $f: \alpha\mapsto\Tr(F(\alpha)^z)$ is differentiable at $\alpha=1$ with derivative $\Tr(z(F(1))^{z-1}F'(1))$.
Now, by l'H\^{o}pital's rule, we have 
$$\lim_{\alpha\to1} D_{\alpha,z}(\rho\|\sigma)=\lim_{\alpha\to1}\frac{\log f(\alpha)}{\alpha-1}=\lim_{\alpha\to1}\frac{f'(\alpha)}{f(\alpha)}=f'(1),$$ 
and it remains to compute $f'(1)$. We have 
$$F'(1)=\frac{1}{z}\left(-\frac{1}{2}(\log\sigma)\rho^{\frac{1}{z}}+(\log\rho)\rho^{\frac{1}{z}}-\frac{1}{2}\rho^{\frac{1}{z}}\log\sigma\right). $$ Hence, Lemma~\ref{29} yields
\begin{eqnarray*}
f'(1)
&=&\Tr(z(F(1))^{z-1}F'(1))\\
&=&\Tr\left(\rho^{\frac{z-1}{z}}\left(-\frac{1}{2}(\log\sigma)\rho^{\frac{1}{z}}+(\log\rho)\rho^{\frac{1}{z}}-\frac{1}{2}\rho^{\frac{1}{z}}\log\sigma\right)\right) ,
\end{eqnarray*}
and the latter term evaluates to $\Tr(\rho(\log\rho-\log\sigma))$, concluding the proof.
\qed
\end{proof}

We are now ready to prove Theorem~\ref{31}.

\begin{proof}[Theorem~\ref{31}]
Since $g(1)\neq0$, we may assume that $g(x)\neq0$ for all $x\in J$, shrinking $J$ if necessary. As usual, define $F(\alpha,z)= \sigma^{\frac{1-\alpha}{2z}}\rho^{\frac{\alpha}{z}}\sigma^{\frac{1-\alpha}{2z}}$, and, moreover, $H(\alpha)=F(\alpha,g(\alpha))^{g(\alpha)}$. Then, 
$$H'(\alpha) = \frac{\partial}{\partial\alpha} F(\alpha,z)^z + g'(\alpha) \frac{\partial}{\partial z} F(\alpha,z)^z \bigg|_{z=g(\alpha)}, $$ which implies that $H$ is continuously differentiable. By employing a similar argument as in Proposition \ref{20}, we find
\begin{eqnarray*}
\lim_{\alpha\to1}D_{\alpha,g(\alpha)}(\rho\|\sigma)
&=& \lim_{\alpha\to1}\Tr(H'(\alpha))\\
&=&\lim_{\alpha\to1}\Tr\left(\frac{\partial}{\partial\alpha} F(\alpha,z)^z + g'(\alpha) \frac{\partial}{\partial z} F(\alpha,z)^z\bigg|_{z=g(\alpha)}\right)\\
&=&\lim_{\alpha\to1} \frac{\partial}{\partial\alpha} \Tr\left( F(\alpha,z)^z \right) +
g'(1) \lim_{\alpha\to1}\frac{\partial}{\partial z} \Tr\left(F(\alpha,z)^z\right) \bigg|_{z=g(\alpha)}\\ 
&=& D(\rho\|\sigma),
\end{eqnarray*} 
where we employed Proposition~\ref{30} in the last step.\qed
\end{proof}

\subsection{Differentiability of $\alpha$-$z$ Divergence as $\alpha\to 1$}

The main purpose of this section is to prove Theorem~\ref{newthm}, which follows from the following two propositions.

\begin{proposition}\label{33}
Let $\rho\in \mathcal{D}(\mathcal{H})$ and $\sigma\in \mathcal{P}(\mathcal{H})$ with $\sigma\gg\rho$. Then $D_{\alpha,1}(\rho\|\sigma)$ is differentiable at $\alpha=1$, and we have $$\frac{\rm d}{{\rm d} \alpha}D_{\alpha,1}(\rho\|\sigma) \bigg|_{\alpha=1}=\frac{V(\rho\|\sigma)}{2}.$$
\end{proposition} 

\begin{proof}
Let us define $f(\alpha)=\Tr\big(\sigma^{\frac{1-\alpha}{2}}\rho^{\alpha}\sigma^{\frac{1-\alpha}{2}}\big)=\Tr\left(\rho^{\alpha}\sigma^{1-\alpha}\right)$ for all $\alpha\in\mathbb{R}$. Then it is easy to see that $f$ is non-zero and infinitely differentiable everywhere on $\mathbb{R}$. This implies that the derivative of $D_{\alpha,1}(\rho\|\sigma)=\frac{\log f(\alpha)}{\alpha-1}$ exists for all $\alpha\in\mathbb{R}\setminus\{1\}$, and is equal to 
$$\frac{\rm d}{{\rm d}\alpha} D_{\alpha,1}(\rho\|\sigma) = \frac{(\alpha-1)\frac{{\rm d}}{{\rm d}\alpha}(\log f(\alpha))-\log f(\alpha)}{(\alpha-1)^2}.$$ 
By l'H\^{o}pital's rule, it suffices to show that $\lim_{\alpha\to1}\frac{{\rm d}}{{\rm d}\alpha}D_{\alpha,1}(\rho\|\sigma)$ exists, and is equal to $\frac12 V(\rho\|\sigma)$. First, since $\lim_{\alpha\to1}\log f(\alpha)=\log f(1)=\log\Tr(\rho)=0,$ and $\lim_{\alpha\to1}\frac{{\rm d}}{{\rm d}\alpha}(\log f(\alpha))$ exists, we have $$\lim_{\alpha\to1}\left[(\alpha-1)\frac{{\rm d}}{{\rm d}\alpha}(\log f(\alpha))-\log f(\alpha)\right]=0.$$ Furthermore, we find that 
\begin{eqnarray*}
&\,&\frac{{\rm d}}{{\rm d}\alpha}\left((\alpha-1)\frac{{\rm d}}{{\rm d}\alpha}(\log f(\alpha))-\log f(\alpha)\right)\\\textsc{}
&=&\frac{{\rm d}}{{\rm d}\alpha}(\log f(\alpha))+(\alpha-1)\frac{{\rm d}^2}{{\rm d}^2\alpha}(\log f(\alpha))-\frac{{\rm d}}{{\rm d}\alpha}(\log f(\alpha))\\
&=&(\alpha-1)\frac{{\rm d}^2}{{\rm d}^2\alpha}(\log f(\alpha)).
\end{eqnarray*}
This implies that 
\begin{eqnarray*}
\frac{\rm d}{{\rm d} \alpha}D_{\alpha,1}(\rho\|\sigma) \bigg|_{\alpha=1} &=&
\lim_{\alpha\to1}\frac{(\alpha-1)\frac{d}{d\alpha}(\log f(\alpha))-\log f(\alpha)}{(\alpha-1)^2} \\
&=&\lim_{\alpha\to1}\frac{\frac{{\rm d}}{{\rm d}\alpha}\left((\alpha-1)\frac{{\rm d}}{{\rm d}\alpha}(\log f(\alpha))-\log f(\alpha)\right)}{\frac{{\rm d}}{{\rm d}\alpha}\left((\alpha-1)^2\right)}\\
&=&\lim_{\alpha\to1}\frac12 \frac{{\rm d}^2}{{\rm d}^2\alpha}(\log f(\alpha)) \,.
\end{eqnarray*}
Hence, we find
\begin{eqnarray*}
\frac{\rm d}{{\rm d} \alpha}D_{\alpha,1}(\rho\|\sigma) \bigg|_{\alpha=1} &=& \frac{f(1)f''(1)-(f'(1))^2}{2(f(1))^2} = \frac{f''(1)-(f'(1))^2}{2}.
\end{eqnarray*}
Now, we note that $f'(\alpha)
=\Tr\left(\rho^{\alpha}\sigma^{1-\alpha}(\log\rho-\log\sigma)\right)$, and 
\begin{eqnarray*}
f''(\alpha)
&=&\Tr\left(\rho^{\alpha}\log(\rho)\sigma^{1-\alpha}(\log\rho-\log\sigma)-\rho^{\alpha}\log(\sigma)\sigma^{1-\alpha}\right)\\
&=&\Tr\left(\rho^{\alpha}(\log\rho-\log\sigma)\sigma^{1-\alpha}(\log\rho-\log\sigma)\right).
\end{eqnarray*} 
This implies that $f'(1)=\Tr\left(\rho(\log\rho-\log\sigma)\right)$ and $f''(1)=\Tr\left(\rho(\log\rho-\log\sigma)^2\right)$, concluding the proof.
\qed
\end{proof}

\begin{proposition}\label{34}
Let $\rho\in \mathcal{D}(\mathcal{H})$ and $\sigma\in \mathcal{P}(\mathcal{H})$ with $\sigma\gg\rho$. Then $D_{\alpha,\alpha}(\rho\|\sigma)$ is differentiable at $\alpha=1$, and we have $$\frac{\rm d}{{\rm d} \alpha} D_{\alpha,\alpha}(\rho\|\sigma) \bigg|_{\alpha=1}=\frac{V(\rho\|\sigma)}{2} .$$ 
\end{proposition} 

\begin{proof}
Let us constrain ourselves to the line $\alpha=z$. Then we have 
\begin{eqnarray*}
\frac{\rm d}{{\rm d}\alpha}D_{\alpha,\alpha}(\rho\|\sigma) &=& \frac{\partial}{\partial\alpha}D_{\alpha,z}(\rho\|\sigma)+\frac{\rm{d} z}{\rm{d}\alpha}\frac{\partial}{\partial z}D_{\alpha,z}(\rho\|\sigma)\\
&=& \frac{\partial}{\partial\alpha}D_{\alpha,z}(\rho\|\sigma)+\frac{\partial}{\partial z}D_{\alpha,z}(\rho\|\sigma).
\end{eqnarray*}
The desired then follows since at $\alpha=z=1$, we have 
\begin{eqnarray*}
\frac{\partial}{\partial\alpha}D_{\alpha,z}(\rho\|\sigma) &=& \frac{\partial}{\partial\alpha}D_{\alpha,1}(\rho\|\sigma) \\
&=& \frac{\Tr\left(\rho(\log\rho-\log\sigma)^2\right)-(\Tr\left(\rho(\log\rho-\log\sigma)\right))^2}{2},
\end{eqnarray*}
and
$\frac{\partial}{\partial z}D_{\alpha,z}(\rho\|\sigma)=\frac{\partial}{\partial z}D_{1,z}(\rho\|\sigma)=0$.
\qed
\end{proof}

\begin{acknowledgements}
 {The results in this paper were achieved as part of a final year project at the School of Computing, Department of Computer Science at the National University of Singapore under the supervision of Prof.~Stephanie Wehner, and we want to thank her for discussions and for providing a conducive research environment at the Centre for Quantum Technologies. We would like an anonymous referee for providing us with Example 1.
 SL would also like to thank Jedrzej Kaniewski, Patrick Coles, and Mischa Woods for discussions leading to the results in this project. Finally, he would like to thank Nelly Ng, his tutor in the Introduction to Information Theory (CS3236) class, for getting him interested in information theory. SL acknowledges support from the Agency for Science, Technology and Research (A*STAR). MT is funded by the Ministry of Education (MOE) and National Research Foundation Singapore, as well as MOE Tier 3 Grant “Random numbers from quantum processes” (MOE2012-T3-1-009).}
\end{acknowledgements}

\end{document}